\documentclass{article}

\usepackage{paper107}

\title{Automating proof search when equality is a logical connective}

\author{
  Kaustuv Chaudhuri \& Arunava Gantait \& Dale Miller
  \\
  \smaller[1]{Inria \& LIX, Institut Polytechnique Paris, France}
}

\begin{document}
\maketitle

\begin{abstract}
Treating syntactic equality as a logical connective---governed by
left- and right-introduction rules within the sequent
calculus---offers an elegant and powerful approach to term
identity. This treatment of equality allows for the derivation of core
mathematical principles, such as Peano's axioms (excluding induction),
and serves as a foundation for the Abella interactive proof
assistant. However, integrating this equality into automated proof
search remains challenging. We present a proof search procedure that
extends unification to handle the complexities of quantifier
alternation and equations that occur in both positive and negative
occurrences. While established logical frameworks such as \lP and LF
lack direct support for this kind of equality, our procedure enables a
lightweight logical framework that addresses this gap. Our system
enables unification-aware proof search across a diverse range of
first-order sequent calculi that can directly use this form of
equality.
\end{abstract}

\section{Introduction}
\label{sec:intro}

Logical frameworks~\cite{pfenning01handbook} typically employ a term language based
on the $\lambda$-calculus equipped with a notion of equality (written $=$)
derived from $\alpha\beta\eta$-conversion.
Implementing such a logical framework requires the ability to reason about this equality
in the presence of unknowns, i.e., over open terms.
For instance, \emph{unification problems} can be represented as a finite list of
(disagreement) pairs of terms $[\tup{t_1,s_1},\dotsc\tup{t_n,s_n}]$, where all
free variables are contained in the list $\bar x = x_1,\dotsc, x_m$.
A solution to such a unification problem is a substitution $\theta$ for the
variables $\bar x$ such that $t_i\theta=s_i\theta$, for every
$i\in\{1,\dotsc,n\}$.
Such a unification problem can therefore be formally depicted as the logical formula:
\begin{gather*}
  \exists x_1.\, \dotsm \exists x_m.\,
  \mleft( t_1 = s_1 \land \dotsm \land t_n = s_n \mright).
\end{gather*}

In logical frameworks supporting \emph{proof search}, such as
\lP~\cite{miller12proghol} and Twelf~\cite{pfenning99cade}, these problems are
generalized to include two kinds of variables: \emph{logic variables}, which
arise from existential quantification and must be instantiated with terms, and
\emph{eigenvariables} (aka \emph{universal variables} and \emph{parameters}),
which arise from universal quantification and are treated as pairwise distinct.
These generalized problems take the form:
\begin{gather*}
  \Qscr_1 x_1.\, \dotsm \Qscr_m x_m.\,
  \mleft( t_1 = s_1 \land \dotsm \land t_n=s_n \mright)
\end{gather*}
where each $\Qscr_i$ is either $\forall$ or $\exists$.
Frameworks for solving unification problems of this ``mixed prefix'' form were
developed in the early
1990s~\cite{miller91jlc,miller92jsc,nadathur93jar,pfenning91lf}.

In the aforementioned frameworks, equality arises during the process of logical
reasoning within the framework, but equality is not itself reified as a logical
connective.
However, Girard~\cite{girard92mail} and
Schroeder-Heister~\cite{schroeder-heister93lics} both proposed, in the early
1990s, a proof-theoretically elegant approach to first-order term equality as a
logical connective, in the sense that it was given left- and right-introduction
rules in a sequent calculus formulation that worked well with cut-elimination.
This treatment was subsequently extended to deal with simply-typed
$\lambda$-terms~\cite{mcdowell00tcs,miller05tocl} and was implemented in the
Bedwyr model checker~\cite{baelde07cade} and the Abella proof
assistant~\cite{baelde14jfr}.
Motivated by the automation of proof search in Abella, we find it necessary to
address unification problems that are more expressive than standard mixed
prefix unification problems.
We develop this richer unification setting in the present paper.

When equality is treated as a logical connective, each conjunct in the
mixed prefix unification problem may involve a collection of equality
antecedents, so the problems now have the following shape:
\begin{gather*}
  \Qscr_1 x_1.\, \dotsm\ \Qscr_m x_m.\,
  \bigwedge_{i \in \{1, \dotsc, n\}}
  \mleft(t_{i,1} = s_{i,1} \imp \dotsm \imp
  t_{i,k_i} = s_{i, k_i} \imp
  t_i = s_i \mright).
\end{gather*}
We also must restrict the provability of such formulas to intuitionistic (and
not classical) logic.
In order to increase the utility of our framework, we also permit atomic
formulas defined via logic-programming-style \emph{definitions}.
We use these to define derivability of essentially any first-order \emph{object
logic} by defining provability in some sequent calculus via a recursive
definition, as detailed in Section~\ref{sec:lwlf}.
The inclusion of atomic formulas requires us to consider 
search problems of the following shape: 
\begin{gather*}
  \Qscr_1 x_1.\, \dotsm \Qscr_m x_m.\,
  \bigwedge_{i \in \{1, \dotsc, n\}}
  \mleft(t_{i,1} = s_{i,1} \imp \dotsm \imp
  t_{i,k_i} = s_{i, k_i} \imp
  H_i \mright),
\end{gather*}
where each $H_i$ is either an equality or a (defined) atom.
Automating problems of this shape is a significant challenge.
Derivability for such formulas is undecidable even when restricted to
first-order terms and empty definitions.
Furthermore, reasoning about equality in the heads $H_i$ necessitates
backtracking search (e.g. via Huet's algorithm~\cite{huet75tcs}), which is only
a semi-decision procedure.
Due to the mixed prefix, neither the positive nor the negative occurrences
of equalities can be fully solved in isolation; instead, partial solutions must
be communicated between the conjuncts of the problem, so that a conjunct that is
otherwise stuck may continue to ensure progress.

In this paper, we introduce a proof search procedure designed to
manage complex unification scenarios involving quantifier alternation and
equations in both positive and negative occurrences. We extend this procedure to
incorporate atomic formulas defined via logic programming-style
specifications. This extension yields SLIM (Specification Logic based on an
Intuitionistic Metalogic), a lightweight logical framework developed to provide
the direct support for equality often absent in established systems. Finally, we
present a sound and relatively complete non-deterministic algorithm for proof
search in SLIM sequents, formalized as a state transition system. We conclude by
discussing how SLIM can be leveraged to provide the Abella system with more
robust automated proof search.

\section{The \eqLJomega logic and its proof system}

This section will present an intuitionistic logic $\eqLJomega$ that is a subset
of Church's Simple Theory of Types~\cite{church40} and provide it with a proof
system based on Gentzen's \LJ proof system~\cite{gentzen35}.
The subscript $\omega$ in the name \eqLJomega indicates that typed quantified
variables can be of arbitrary order.
Eventually we will identify a subset of this logic that we shall consider as our
lightweight logical framework: in that subset, the order of typed quantified
variables will be limited to be 0 or 1.

\subsection{Types, terms, and formulas}

Types (written $\sigma$ or $\tau$) consist of \emph{primitive types} including
the primitive type $o$ of \emph{formulas} (following Church's convention) and
other primitive types that are declared when we have a need for them; and
\emph{arrow types} $\sigma \ra \tau$ with $\ra$ being right-associative.
The \emph{order} of a type $\sigma$, written $\ord(\sigma)$, is defined to be
$0$ if $\sigma$ is a primitive type, and
$\ord(\sigma \ra \tau) = \max(\ord(\sigma) + 1, \ord(\tau))$.
For a type $\sigma = \sigma_1 \ra \dotsm \ra \sigma_n \ra \tau$ where $\tau$ is
a primitive type, we say that $\sigma_1, \dotsc, \sigma_n$ are the
\emph{argument types} of $\sigma$ and $\tau$ is its \emph{target type}.
Note that if $\ord(\sigma) \le 1$, then every argument type of $\sigma$ is a
primitive type.
These types are inhabited by $\lambda$-terms with the usual syntax and typing
rules, which we elide in this paper.

Following Church, logic is introduced to this simple type theory by declaring
the following logical constants, together with their associated types:
\begin{gather*}
  \true, \false : o,
  \quad
  {\wedge}, {\vee}, {\imp} : o \ra o \ra o,
  \quad
  % \false\colon o,~
  % \vee\colon o\ra o\ra o,~\imp\colon o\ra o\ra o,~\\
  =_\sigma\colon \sigma\ra\sigma\ra o,
  \quad
  {\forall_\sigma}, {\exists_\sigma} : (\sigma \ra o) \ra o.
\end{gather*}
We will write $\wedge$, $\vee$, $\imp$, and $=_\sigma$ infix, with the usual
conventions.
We will also abbreviate $\forall_\sigma (\lambda x.\,F)$ and
$\exists_\sigma (\lambda x.\,F)$ as $\forall_\sigma x.\, F$ and
$\exists_\sigma x.\, F$ respectively.
Finally, we will drop the type-suffixes on quantifiers and equality if they
are clear from context or unimportant to the exposition.
These logical constants will be given a meaning in terms of inference rules in a
sequent calculus in the next subsection.

The logical constants $\forall_\sigma$, $\exists_\sigma$, and $=_\sigma$ will
have a restriction that the type $\sigma$ cannot contain $o$.
Thus, we will obtain a \emph{first-order logic over higher-order
  $\lambda$-terms}.
Proof-theoretic meta-theorems about this logic can be proved by elementary
induction in the style of Gentzen's cut-elimination theorem~\cite{gentzen35},
and will not require stronger reasoning principles such as reducibility
candidates~\cite{girard89book}.

\subsection{The sequent calculus for \eqLJomega}

\begin{figure}[t]
\begin{gather*}
  \linfer[\init]{
    \Pref\twoseq{A, \Gamma}{A}
  }{}
  \qquad
  \linfer[{\true}R]{
    \Pref\twoseq{\Gamma}{\true}
  }{ }
  \qquad
  \linfer[{\true}L]{
    \Pref\twoseq{\true, \Gamma}{C}
  }{
    \Pref\twoseq{\Gamma}{C}
  }
  \\[1ex]
  \infer[{\wedge}R]{
    \Pref\twoseq{\Gamma}{B\wedge C}
  }{
    \Pref\twoseq{\Gamma}{B}
    &
    \Pref\twoseq{\Gamma}{C}
  }
  \qquad
  \infer[{\wedge}L]{
    \Pref\twoseq{A \wedge B, \Gamma}{C}
  }{
    \Pref\twoseq{A, B, \Gamma}{C}
  }
  \\[1ex]
  \infer[{\vee}R_i, i \in \{1, 2\}]{
    \Pref\twoseq{\Gamma}{C_1 \vee C_2}
  }{
    \Pref\twoseq{\Gamma}{C_i}
  }
  \qquad
  \infer[{\vee}L]{
    \Pref\twoseq{A \vee B, \Gamma}{C}
  }{
    \Pref\twoseq{A, \Gamma}{C}
    &
    \Pref\twoseq{B, \Gamma}{C}
  }
  \\[1ex]
  \text{(no ${\false}R$)}
  \qquad
  \linfer[{\false}L]{
    \Pref\twoseq{\false, \Gamma}{C}
  }{ }
  \\[1ex]
  \infer[{\imp}R]{
    \Pref\twoseq{\Gamma}{B \imp C}
  }{
    \Pref\twoseq{B, \Gamma}{C}
  }
  \qquad
  \infer[{\imp}L]{
    \Pref\twoseq{A \imp B, \Gamma}{C}
  }{
    \Pref\twoseq{A \imp B, \Gamma}{A}
    &
    \Pref\twoseq{B, A \imp B, \Gamma}{C}
  }
  \\[1ex]
  \infer[{\forall}R]{
    \Pref\twoseq{\Gamma}{\forall_\sigma x.\, C}
  }{
    \Prefe{y{:}\sigma}\twoseq{\Gamma}{C[y/x]}
  }
  \qquad
  \infer[{\forall}L]{
    \Pref\twoseq{\forall_\sigma x.\, A, \Gamma}{C}
  }{
    \Sigma \vdash t : \sigma
    &
    \Pref\twoseq{A[t/x], \forall_\sigma x.\, A, \Gamma}{C}
  }
  \\[1ex]
  \infer[{\exists}R]{
    \Pref\twoseq{\Gamma}{\exists_\sigma x.\, C}
  }{
    \Sigma \vdash t : \sigma
    &
    \Pref\twoseq{\Gamma}{C[t/x]}
  }
  \qquad
  \infer[{\exists}L]{
    \Pref\twoseq{\exists_\sigma x.\, A, \Gamma}{C}
  }{
    \Prefe{y{:}\sigma}\twoseq{A[y/x], \Gamma}{C}
  }
  \\
  \infer[{=}R]{
    \Pref\twoseq{\Gamma}{t = t}
  }{ }
  \qquad
  \infer[{=}L]{
    \Pref\twoseq{s = t, \Gamma}{C}
  }{
    \mleft\{
    \Pref[\Sigma\theta]\twoseq{\Gamma\theta}{C\theta}
    \sep
    \theta \in \csu{s, t}
    \mright\}
  }
\end{gather*}
\caption{The inference rules for the \eqLJomega proof system}
\label{fig:eqLJ}
\end{figure}

The system \eqLJomega is a standard sequent calculus in the style of Gentzen's
\LJ~\cite{gentzen35} that gives the meaning for the logical constants defined
above, with the rules given in figure~\ref{fig:eqLJ}.
The main differences from \LJ are as follows:
\begin{itemize}
\item{} Sequents in \eqLJomega will be of the form $\Pref\twoseq{\Gamma}{C}$,
  where $\Gamma$ (the assumptions) is a set of formulas, $C$ is a formula, and
  $\Sigma$ is explained below.
  The context $\Gamma$ of assumptions is treated as a set instead of a multiset
  or a list.
  Thus, the structural rules of weakening, contraction, and exchange are
  admissible and therefore unnecessary.
\item The typed eigenvariables introduced in the ${\forall}R$ and ${\exists}L$
  rules are collected in a typing context $\Sigma$ prefix.
  For such typing contexts, we assume the existence of a typing judgment
  $\Sigma \vdash t : \sigma$ for asserting that the term $t$ has type $\sigma$,
  with the usual rules (omitted here).
  This judgment is used to type-check the witness terms used in the ${\forall}L$
  and ${\exists}R$ rules.
\end{itemize}
The rules ${=}L$ and ${=}R$ for equality were inspired
by~\cite{girard92mail,schroeder-heister93lics}, but were generalized
in~\cite{mcdowell00tcs} in order to work also for a logic in which terms are
simply typed $\lambda$-terms.
While ${=}R$ is simple and natural, the ${=}L$ rule is more subtle.
In that rule, $\csu{s, t}$ denotes the \emph{complete set of unifiers} for $s$
and $t$, $\Sigma\theta$ stands for the typing context for the image under
$\theta$ of the eigenvariables in $\Sigma$, and $\Gamma\theta$ stands for the
pointwise image under $\theta$ of each element of $\Gamma$.
In the fully general setting of the simply typed $\lambda$-calculus, this set
$\csu{s, t}$ can be empty (i.e., $s$ and $t$ are not unifiable), non-empty and
finite, or infinite.
Thus, if $s$ and $t$ are not unifiable, the sequent
$\Pref\twoseq{s=t,\Gamma}{C}$ is proved since the set of premises for ${=}L$ is
empty.

On the other hand, we might need to consider instances of ${=}L$ with infinitely
many premises.
Thus, proofs in this rich setting can be difficult to study, requiring
transfinite induction to establish some basic definitions and meta-theory
results.
We will be avoiding these issues soon since we will select a small subset of
this logic and proof system in which we know that ${=}L$ will only have 0 or 1
premises.

The proof theory of all of \eqLJomega has been developed
in~\cite{mcdowell00tcs}, where a larger logic $\FOLDN$ containing \eqLJomega is
shown to admit a cut-elimination theorem.

\begin{example}\label{ex:equiv=}
  Given $\Sigma \vdash C : o$, we say that $C$ is provable if the sequent
  $\Pref\twoseq{\cdot}{C}$ is derivable in \eqLJomega.
  The following closed formulas, which state that equality is an equivalence,
  are all provable in \eqLJomega:
  $\forall x.\, x = x$, $\forall x.\, \forall y.\, \mleft(x = y \imp y = x\mright)$,
  $\forall x.\, \forall y.\, \forall u.\, \mleft(x = y \imp y = u\imp x = u\mright)$.
\end{example}

\begin{example}\label{ex:succ}
  Assume that $\nat$ is a primitive type and that $z$ (zero) has type $\nat$
  and $s$ (successor) has type $\nat\ra\nat$.  The formulas
  $\forall x.\, \forall y.\, \mleft(s\ x = s\ y \imp x = y\mright)$ and
  $\forall x.\, \mleft(s~x = z \imp \false\mright)$, 
  which state that equality is injective and a successor is never zero, are
  provable in \eqLJomega.
\end{example}

By putting these two examples together, we see that the proof system \eqLJomega forms
a basis for (Heyting) arithmetic since it proves the first four Peano axioms.
The fifth axiom (involving induction) is not part of this system.

In subsets of this logic that do not contain equality (such as Gentzen's
\LJ~\cite{gentzen35}), the roles of eigenvariables of primitive type and
constructors of arity 0 are essentially interchangeable.
For example, assume that $a, b : i$ are two distinct constants for some
primitive type $i$.
Then the $\Sigma\cup\{a,b\}$-formula $P$ has an \eqLJomega-proof if and only if
$\forall x.\, \forall y.\, P[x/a, y/b]$ has an \eqLJomega-proof.
Given $a, b{:}i, \Sigma \vdash P : o$, it is the case that if $P$ is provable in
\eqLJomega then $c{:}i, \Sigma \vdash P[c/a, c/b] : o$ and $P[c/a, c/b]$ is
provable in \eqLJomega.
In other words, constants can be seen as playing a generic role.
In the full logic including equality, this interchangeability no longer holds.
For example, the formula $a = b \supset\false$ is provable but replacing both
$a$ and $b$ with, say, $c$ yields a formula that is not provable.

\section{SLIM: a light-weight logical framework}
\label{sec:lwlf}

Since our main goal with using \eqLJomega is to specify sequent calculus proof
systems for first-order logics (that may use equality), we can work with a rather
small subset of \eqLJomega.
In particular, we will dispense with disjunction ($\vee$) entirely, and restrict
the antecedents of implications to equalities.
This smaller logic will nevertheless be sufficient to specify a rich set of
proof systems for a wide range of object logics, as we will illustrate later.

Our light-weight logical framework, SLIM (for ``Specification Logic based on an
Intuitionistic Metalogic'') is motivated by higher-order logic programming in
the style of \lP~\cite{miller12proghol}.
Following the terminology of logic programming, we define two classes of
\eqLJomega formulas: goal formulas ($G$) and definite clauses ($D$).
\begin{align*}
G := &\ \true \sep \false \sep A \sep G_1\wedge G_2\sep
        \exists_\tau x.\, G\sep \forall_\sigma x.\, G \sep t = s\sep t = s \imp G \\
D := &\ G\imp A \sep \forall_\tau x.\, D
\end{align*}
Here, $t$ and $s$ denote $\lambda$-terms as before, and $A$ denotes an atomic
formula, which is any formula of the form $p\ t_1\ \dotsm\ t_n$ where $p$ is a
predicate (i.e., a non-logical constant whose target type is $o$) and
$t_1, \dotsc, t_n$ are terms.
We also make the following restrictions on the types of the quantifiers: in the
goal formula $\forall_\sigma x.\, G$, $\sigma$ must be a primitive type (order
0) while in the goal formula $\exists_\tau.\, G$ or the definite clause
$\forall_\tau x.\, D$ the type $\tau$ is of order 0 or 1.
In order to reason in SLIM, we will primarily consider \eqLJomega sequents of
the form $\Pref\twoseq{\Pscr}{G}$ where $G$ is any goal formula and the context
$\Pscr$ contains only \emph{closed} definite clauses (the \emph{logic program}).
Given the typing restrictions in SLIM goal formulas, we will only need to
consider signatures, $\Sigma$, that map eigenvariables to 
primitive types.

Note that even though SLIM has few logical building blocks, its proof
search problem is undecidable.
This is the case even for the fragment with an empty logic program and no atoms
by reduction to Hilbert's 10th problem~\cite{miller22amai}.
(Elaborated further in
Appendix~\ref{apx:empty-lp}%
.)

\subsection{Logic program specifications of proof systems}
\label{sec:systems}

While SLIM is a weak framework, the ability to attach logic specifications to it
makes it possible to specify a large variety of sequent calculi.
Indeed, specifying sequent calculus proof systems using definite clauses is easy
to do using techniques dating back to the late
1980s~\cite{felty88cade,miller12proghol}.
For example, let \eqLJone be the first-order logic that results from requiring
all quantification within formulas to be of primitive types, and now consider
\eqLJone as an \emph{object logic}.  In that setting, we could interpret the
${\wedge}R$ rule as a rule on a new atomic formula built using the predicate
$\lsti|and|$, whose inference rule looks like this:
\begin{gather*}
  \infer{
    \Gamma\vdash \lsti|and|\ B\ C
  }{
    \Gamma\vdash B
    &
    \Gamma\vdash C
  }
\end{gather*}
where $\Gamma$ is now interpreted not as a set of formulas of type $o$ but
rather as a container (say a list) of terms of type \lsti|fm| (say).
The rule itself is easy to write in \lP using primitive types \lsti|fm| and
\lsti|fmlist| for formulas and contexts, a constant
$\lsti|and| : \lsti|fm| \ra \lsti|fm| \ra \lsti|fm|$, a non-logical constant
$\lsti|seq| : \lsti|fmlist| \ra \lsti|fm| \ra o$, and the following clause
(written using \lP syntax):
\begin{lstlisting}
  seq Gamma (and B C) :- seq Gamma B, seq Gamma C.
\end{lstlisting}
We therefore use logical formulas at two levels.
The object logic eventually contains formulas of type \lsti|fm| built using
constants such as \lsti{and}.
The \emph{framework logic}, which is the role that SLIM plays, is the logic that
represents the clause above as the following definite clause:
\begin{gather*}
  \forall \Gamma.\, \forall B.\, \forall C.\,
  \mleft((\lsti{seq}\ \Gamma\ B \wedge
  \lsti{seq}\ \Gamma\ C)
  \imp
  \lsti{seq}\ \Gamma\ (\lsti{and}\ B\ C)\mright).
\end{gather*}

Logical frameworks based on the intuitionistic logic of hereditary Harrop
formulas~\cite{miller91apal,miller12proghol} or on dependently typed
$\lambda$-calculus~\cite{harper93jacm,pfenning99cade}, can support this kind of
specification of object logics -- \emph{except} when it comes to equality
itself, which these frameworks do not provide as a usable connective at the
framework-level.
Thus, to define \eqLJone as an object logic requires a framework such as
SLIM.\footnote{%
  For an example of another and rather different proof system, see
  Appendix~\ref{apx:linear}.
}

\begin{figure}[tp]
\begin{lstlisting}[language=lprolog,basicstyle={\ttfamily}]
kind tm, atm, fm      type.             % Terms, atoms, formulas
type atom             atm -> fm.        % Atoms
type tt, ff           fm.               % Units
type or, and, imp     fm -> fm -> fm.   % Binary connectives
type forall, exists   (tm -> fm) -> fm. % Quantifiers
type eq               tm -> tm -> fm.   % Equality

kind fmlist  type.
type nil     fmlist.
type ::      fm -> fmlist -> fmlist. % written infix

type mnr     fm -> fmlist -> fmlist -> oo. % "member and rest"
  % mnr X G D means G = X, D (modulo ordering)

mnr X (X :: L) L.
mnr X (Y :: L) (Y :: K) :- mnr X L K.

type seq     fmlist -> fm -> oo.

% initial
seq Gam (atom A) :- mnr (atom A) Gam _.

% right rules
seq Gam tt.
seq Gam (eq T T).
seq Gam (and B C)  :- seq Gam B, seq Gam C.
seq Gam (or  B C)  :- seq Gam B.
seq Gam (or  B C)  :- seq Gam C.
seq Gam (imp B C)  :- seq (B :: Gam) C.
seq Gam (exists C) :- sigma t\ seq Gam (C t).
seq Gam (forall C) :-    pi y\ seq Gam (C y).

% left rules
seq Gam C :- mnr ff Gam _.
seq Gam C :- mnr (eq S T) Gam Del,   (S = T) => seq Del C.
seq Gam C :- mnr (and A B) Gam Del,  seq (A :: B :: Del) C.
seq Gam C :- mnr tt Gam Del,         seq Del C.
seq Gam C :- mnr (or A B) Gam Del,   seq (A :: Del) C, seq (B :: Del) C.
seq Gam C :- mnr ff Gam _.
seq Gam C :- mnr (imp A B) Gam _,    seq Gam A, seq (B :: Gam) C.
seq Gam C :- mnr (forall A) Gam _,   sigma t\ seq (A t :: Gam) C.
seq Gam C :- mnr (exists A) Gam Del,    pi y\ seq (A y :: Del) C.
\end{lstlisting}
\caption{A faithful specification of \eqLJone as an object logic (compare with
  Figure~\ref{fig:eqLJ})} 
\label{fig:eqlj-mod}
\end{figure}

\begin{example}[\eqLJone as an object logic]
  The type declarations (using \lP syntax) in Figure~\ref{fig:eqlj-mod} declare
  the formulas of an object logic which are here treated as terms of type
  \lsti{fm}.
  For the connective $=_\sigma$ and quantifiers $\forall_\sigma$ and
  $\exists_\sigma$ of \eqLJomega, we provide examples for the single primitive
  type \lsti|tm|.\footnote{%
    In \lP these could be written using polymorphic constants, but we avoid the
    complication of polymorphic types in this paper.%
  }
  Figure~\ref{fig:eqlj-mod} lists all the clauses specifying the rules of \eqLJomega
  in Figure~\ref{fig:eqLJ} using \lP.
  For instance, the clause corresponding to ${\exists}R$ can be seen as the following
  definite clause:
  \begin{gather*}
    \forall \Gamma.\, \forall C.\,
    \mleft(\exists t.\, \lsti{seq}\ \Gamma\ (C\ t)\mright)
    \imp \lsti{seq}\ \Gamma\ (\lsti{exists}\ C),
  \end{gather*}
  where $\Gamma$ has type \lsti|fmlist| and $C$ has type
  $\lsti|tm| \ra \lsti|fm|$, both types of order 1 or less.

  Note that this specification represents contexts as lists (the type
  \lsti|fmlist|) and uses an auxiliary predicate \lsti|mnr| to reason about
  membership in this list.
  For instance, the clause for ${=}L$ is the following definite clause is:
  \begin{gather*}
    \forall \Gamma, \Gamma'.\, \forall C.\, \forall S, T.\,
    \mleft(\lsti|mnr|\ (\lsti|eq|\ S\ T)\ \Gamma\ \Gamma'
    \wedge
    \mleft((S = T) \imp \lsti|seq|\ \Gamma'\ C)\mright)
    \mright) \imp
    \lsti|seq|\ \Gamma\ C.
  \end{gather*}
  The \lsti|mnr| predicate is used to remove $\lsti|eq|\ S\ T$ from $\Gamma$
  (yielding $\Gamma'$), and then the SLIM $=$ is used to guard the goal
  $\lsti|seq|\ \Gamma'\ C$, letting $=$ represent \lsti|eq|.
\end{example}

Our goal in the rest of this paper is to describe a non-deterministic algorithm
that provides a sound and complete proof search strategy for SLIM sequents of
the form $\Sigma\Colon\twoseq{\Pscr}{G}$.
Of course a SLIM proof of such a sequent does not necessarily correspond to an
object proof system that may be represented using the SLIM atoms.
Fortunately, there is an easy way to recover the object proofs as well by
enriching predicates such as \lsti|seq| above with an additional argument which
would be the proof term in the object logic; see, for
example,~\cite{felty88cade,chihani17jar}.
Whether or not object proofs are explicitly represented is a choice that is
incorporated into the logic program specification of the object logic.
Indeed, while SLIM sequents may be proved using the sequent calculus, SLIM
itself does not provide any proof terms at the framework level.

\subsection{Solutions}

Let $G$ be a closed goal formula.
For convenience, assume that the bound variables in $G$ are all given distinct
names.
A \emph{$G$-substitution} is a mapping of every occurrence of existential
quantifiers in $G$ to terms such that if $\exists_\tau x$ is mapped to $t$ then
$t$ contains a free variable only if that free variable is universally
quantified by an occurrence of $\forall y$ that has $\exists_\tau x$ in its
scope.
The application of a $G$-substitution $\theta$ to $G$ is the result of
instantiating every subexpression in $G$ of the form $\exists x.\, G'$
with the term associated to it; we write this as $\theta[G]$.

A \emph{solution for $\twoseq{\Pscr}{G}$} is a $G$-substitution such the sequent
$\Pref[\cdot]\twoseq{\Pscr}{\theta[G]}$ is provable in \eqLJomega.
In general, it is undecidable whether or not a given substitution is a solution
for a given sequent, but it is decidable if $G$ has no atomic formulas.

\begin{example}
  Let $i$ be a primitive type, and $a : i$ and $g : i\ra i\ra i$ be constants.
  There are four solutions to the sequent
  \begin{gather*}
    \twoseq{\Pscr}{
      \forall x_1.\, \forall x_2.\, \exists u.\,
      \mleft(x_1 = a \imp x_2 = a \imp u = (g\ a\ a)\mright)
    },
  \end{gather*}
  namely, one solution each when $\exists u$ is instantiated with one of the
  terms $(g\ x_1\ x_2)$, $(g\ a\ x_2)$, $(g\ x_1\ a)$, $(g\ a\ a)$.
  Since this goal formula has no atomic formulas, this result is independent of
  the choice of the logic program $\Pscr$.
\end{example}

\section{State formulas}
\label{sec:state}

We will find it useful to simplify the structure of goal formulas.
A \emph{guarded goal} is a formula of the form
$
  \forall \bar y.\,
  t_1 = s_1 \imp \dotsm \imp t_m = s_m \imp B,
$
where $\bar y$ is a list of variables of primitive type and $B$ is
either an atom, $\false$, or an equality of terms of primitive type.
The formula occurrence $B$ is the \emph{target} of this guarded goal and the
equations $t_1 = s_1$, \dots, $t_m = s_m$ are called its \emph{guarding
  equalities}.
Depending on the form of $B$, we can call such formulas either a \emph{guarded
  atom}, a \emph{guarded equality}, or a \emph{guarded false}.
We drop the quantifier $\forall$ if $\bar y$ is empty, and the implications if
$m = 0$, so a single equality is also a guarded goal.

A \emph{state formula} is a formula of the form $\exists \bar h.\, \true$ or
$\exists \bar h.\, (G_1 \wedge \dotsm \wedge G_n)$ (for $n > 0$), where $\bar h$
is a list of typed variables of order 0 or 1, and each $G_i$ is a guarded goal.
In such formulas, there are no occurrences of existential quantifiers and
conjunctions in the scope of either universal quantifiers or implications.

We will argue that every SLIM goal formula has a corresponding state formula
that preserves solutions.
A key step is a (generalized) notion of \emph{raising}~\cite{miller92jsc}.

\begin{definition}[raising]
  Rewriting the goal formula
  \begin{gather*}
    \forall\bar y.\, t_1 = s_1 \imp \dotsm \imp t_n = s_n \imp \exists x.\, G
  \end{gather*}
  (where $n\ge 0$) to the goal formula
  \begin{gather*}
    \exists h.\, \forall\bar y.\,
    t_1 = s_1 \imp \dotsm \imp t_n = s_n \imp \subst{(h\ \bar y)}{x}{G}
  \end{gather*}
  is called \emph{raising}, where $h\ \bar y$ stands for the iterated
  application of $h$ to $\bar y$.
  In this rewritten formula, the scope of the existential quantification is
  expanded and the type of the existential quantifier is raised from, say,
  $\tau$ to $\sigma_1\ra\dotsm\ra\sigma_n\ra\tau$, where $\sigma_i$ is the
  (primitive) type of $x_i$ for $i \in \{1,\dotsc,n\}$.
  We write $G \raisesto G'$ if $G'$ results from raising $G$, and $\raisesto*$
  for the reflexive-transitive closure of $\raisesto$.
\end{definition}

\begin{example}\label{ex:one}
The goal formula
  \begin{gather*}
    \forall u.\, \forall v.\, \exists x.\,
    \mleft({
      (u =a \imp v = b \imp x = a)
      \wedge
      (u=b\imp x=u)
    }\mright).
  \end{gather*}
  has exactly one solution, namely, the substitution of $u$ for $x$.
  The raised version of this goal formula is
  \begin{gather*}
    \exists h.\, \forall u.\, \forall v.\,
    \mleft({
      (u = a \imp v = b \imp (h\ u\ v) = a)
      \wedge
      (u = b \imp (h\ u\ v) = u)
    }\mright),
  \end{gather*}
  and it has exactly one solution, namely, the substitution of
  $\lambda w_1.\, \lambda w_2.\, w_1$ for $h$.
\end{example}

\begin{theorem}
  [raising preserves solutions]
  If $G_0 \raisesto* G_1$, then the solutions of the sequent
  $\Pref\twoseq{\Pscr}{G_0}$ are in bijection to those of
  $\Pref\twoseq{\Pscr}{G_1}$.
\end{theorem}

\begin{proof}
  By definition a solution must substitute an existential with a term built from
  the universally quantified variables in scope.
  Every such replacement of the form $x \mapsto t[t_1/y_1, \dotsc, t_n/y_n]$,
  where $y_1, \dotsc, y_n$ are the universal variables in scope of $\exists x$,
  has a corresponding replacement of the form
  $x \mapsto (\lambda y_1.\, \dotsm \lambda y_n.\, t)\ t_1\ \dotsm\ t_n$ after
  raising, and vice versa.
\end{proof}

\begin{corollary}[normalization]
  For any closed goal formula $G$, there is a state formula $S$ such that
  the solutions of $\cdot\Colon\twoseq{\Pscr}{G}$ are in bijection to those of
  $\cdot\Colon\twoseq{\Pscr}{S}$.
\end{corollary}

\begin{proof}
  The state formula $S$ (called the \emph{normalization} of $G$) is easily built
  by a combination of raising and using 
  the following equivalences, where $A \equiv B$ is an abbreviation for
  $(A \imp B) \wedge (B \imp A)$.
  \begin{align*}
    (\exists x. P~x)\wedge Q &\equiv \exists x.(P~x\wedge Q)\\
    t=s \imp (G_1\wedge G_2) &\equiv (t=s\imp G_1)\wedge (t=s\imp G_2)\\
    t=s \imp \forall y.~G    &\equiv \forall y.~t=s \imp  G\\
    (\forall y. P~y)\wedge (\forall y. Q y) &\equiv\forall y.(P~y\wedge Q~y)\\
    t =_{\sigma\ra\tau} s & \equiv \forall y_\sigma.~t~y =_\tau s~y
  \end{align*}
  Then we simply observe that logical equivalences preserve solutions and
  raising keeps the solutions in bijection.
\end{proof}

\begin{definition}[rigid subterms]
  A term $s$ is an \emph{immediate rigid subterm} of $t$, written $t > s$, if
  $t$ is of the form $(f\ t_1\ \dotsm\ s\ \dotsm\ t_n)$ where $f$ is a constant.
  A \emph{rigid path} from $t$ to $s$ is a sequence of terms $t_1, \dotsc, t_n$
  ($n > 1$) such that $t_1$ is $t$, $t_n$ is $s$, and $t_i > t_{i + 1}$ for
  every $i \in \{1, \dotsc, n - 1\}$.
  We say that $s$ is a \emph{rigid subterm} of $t$, written $t >^+ s$, if there is a
  rigid path from $t$ to $s$.
\end{definition}

We will now specify a procedure for simplifying normalized state formulas into
\emph{reduced} state formulas that performs as much equality simplification of
guarding equalities as possible.

\begin{definition}[reduction]
  \label{defn:reduction}%
  Given a normalized state formula $S$, we say $S$ \emph{reduces to} the state
  formula $S'$, written $S \redsto S'$, if:
  \begin{itemize}
  \item $S'$ results from $S$ by rewriting some occurrence of an equality
    according to the following rule, then normalizing: if $f$ is a constant, then
    \begin{multline*}
      \forall \bar y.\,
      \mleft(Q_1 \imp \dotsm \imp(f\ s_1 \dotsm s_k) = (f\ t_1 \dotsm t_k) \imp
      \dotsm \imp Q_m \imp B\mright)
      \quad \redsto
      \\
      \forall \bar y.\,
      \mleft(Q_1 \imp \dotsm \imp s_1 = t_1 \imp \dotsm \imp s_k = t_k \imp
      \dotsm \imp Q_m \imp B\mright)
    \end{multline*}
  \item $S'$ results from $S$ by rewriting a guarded goal as follows, then
    normalizing:
    \begin{multline*}
      \forall \bar y.\,
      \mleft(Q_1 \imp \dotsm \imp Q_i\imp y = t \imp Q_{i+1}\imp 
                      \dotsm \imp Q_m \imp B\mright)
      \quad \redsto
      \\
      \forall \bar y_{- y}.\,
      \mleft(Q_1 \imp \dotsm \imp Q_m \imp B\mright)[t/y]
    \end{multline*}
    where $y \in \bar y$, $y$ is not a rigid subterm of $t$, each of the $Q_i$
    is an equality, and $\bar y_{-y}$ stands for the list of variables $\bar y$
    with $y$ removed.
    We also include the symmetric case with the guard $t = y$ instead of $y = t$.

  \item (occurs check) $S'$ results from $S$ by replacing a guarded goal in $S$ 
of the form
    %
%    \begin{gather*}
$
      \forall \bar y.\,
      \mleft(Q_1 \imp \dotsm \imp y = t \imp \dotsm \imp Q_m \imp B\mright)
$
%    \end{gather*}
    %
    where $y \in \bar y$, and $y$ is a rigid subterm of $t$, with the formula $\true$.
    We also include the symmetric case with the guard $t = y$ instead of $y = t$.
  \end{itemize}
  We write $\redsto*$ for the reflexive-transitive closure of $\redsto$.
  A normalized state formula is said to be \emph{reduced} if there is no $S'$
  such that $S \redsto S'$.
\end{definition}

In the rest of this and the next section, we use $S$ as a schematic variable to
denote normalized state formulas.

\begin{proposition}[reduction preserves solutions]
  \label{thm:reduction-preserves}%
  If $S \redsto* S'$, then the solutions of $\Pref\twoseq{\Pscr}{S}$ are in
  bijection with those of $\Pref\twoseq{\Pscr}{S'}$.
\end{proposition}

\begin{proof}
  No existential quantifiers are touched during reduction and the processing steps
  are essentially solving only first-order unification problems.
\end{proof}

\begin{example}
  The state formula in Example~\ref{ex:one}, namely, 
  \begin{gather*}
    \exists h.\,\forall u, v.\,
    \mleft({
      (u=a\imp v=b\imp(h\ u\ v)=a) \wedge
      (u=b\imp(h\ u\ v)=u)
    }\mright)
  \end{gather*}
  is not reduced.
  It can be reduced to
  ${
    \exists h.\, \forall u, v.\,
    \mleft({
      (h\ a\ b)=a \wedge (h\ b\ v)=b
    }\mright)
  }$.
  These two formulas have the same solutions.
  On the other hand, the following formula
  \begin{gather*}
    \exists h.\,
    \mleft({
      (\forall y.\,
      y=(f\ (h\ y)) \imp\false)
      \wedge (h\ a = a)
    }\mright)
  \end{gather*}
  is reduced since $y$ is not a rigid subterm of $(f\ (h\ y))$.
  Note that in isolation, the equation $(h\ a = a)$ has two solutions, namely,
  where $h$ is instantiated with $\lambda w.\, w$ or $\lambda w.\, a$, but only
  the first of these will solve the first guarded goal.
\end{example}

\begin{example}\label{ex:interpolation}
  An \emph{interpolation problem}~\cite{dowek94apal} is a set of
  equations of the form $$(x\ t_1\ \dotsm\ t_n) = b,$$ where
  such that $x$ is an instantiable variable and $t_1, \dotsc, t_n$ and
  $b$ are ground terms.
  If all the instantiable variables in an interpolation problem have types with
  order less than or equal to $n$, we say that this problem is of order $n+1$.
  There is the following immediate correspondence between interpolation problems
  containing one equation (and, hence, one instantiable variable) and goal
  formulas.
  \begin{align}
    \exists x.~(x\ t_1\ \dotsm\ t_n) =&\ b
    \label{eqn:one} \\
    \exists x.\, \forall y_1, \dotsc, y_n.\,
    \bigl(
    y_1 = t_1\imp\dotsm\imp y_n = t_n\imp (x_1\ y_1\ \dotsm\ y_n)
    =&\ b \bigr)
    \label{eqn:two} \\
    \forall y_1, \dotsc, y_n.\, \exists u.\,
    \bigl(
    y_1 = t_1 \imp\dotsm\imp y_n = t_n\imp u
    = &\ b \bigr)
    \label{eqn:three}
  \end{align}
  Formula (\ref{eqn:two}) arises from (\ref{eqn:one}) by applying the
  equivalence $(\forall v.\, v=t \imp B)\equiv \subst{t}{u}{B}$ (assuming that
  $v$ is not free in $t$).
  Formula (\ref{eqn:two}) is related to (\ref{eqn:three}) by raising.
  If (\ref{eqn:one}) is a second-order interpolation problem, then
  (\ref{eqn:three}) contains only first-order quantification.
  Thus, proofs of (\ref{eqn:one}) (\ie, substitutions for $x$) are in one-to-one
  correspondence with proof of (\ref{eqn:three}).
  Of these three formulas, only (\ref{eqn:one}) and (\ref{eqn:two}) are state
  formulas and only (\ref{eqn:one}) is in reduced form.
\end{example}

\section{Reasoning about targets in reduced state formulas}
\label{sec:transition}

State formulas can be reduced (Definition~\ref{defn:reduction}) eagerly since,
by Proposition~\ref{thm:reduction-preserves}, all solutions are preserved.
To solve a state formula, we therefore only need to consider reduced state
formulas.
Reduction also fully handles all reasoning on the guarding equalities, so we
only need to consider the targets of guarded goals.
There are two kinds of reasoning steps that are allowed for such goals.
\begin{itemize}
\item \emph{Backchaining} steps, which make use of a definite clause in the
  program $\Pscr$ to ``unfold'' a guarded atom into the corresponding body of
  the clause.
\item \emph{Unification} steps, which reason about a guarded equality.
\end{itemize}

We will write such steps in the form of a labeled transition system on state
formulas, $\twoseq{\Pscr}{S \transto[\rho] S'}$ where $\Pscr$ is a program, $S$
and $S'$ are reduced state formulas, and $\rho$ is an $S$-substitution.
Generally, $\rho$ is neither a solution for $S$ nor for $S'$: instead, composing
$\rho$ with a solution for $S'$ yields a solution for $S$.

\paragraph{Backchaining steps:}

Pick a goal in $S$ with a guarded atom, i.e., a goal formula of the following
form:
$\forall \bar y.\, \mleft(Q_1 \imp \dotsm \imp Q_m \imp p\ t_1\ \dotsm\
t_n\mright)$ where $p$ is a predicate symbol, the $Q_i$ are the guarding
equalities, and the $t_j$ are terms.
Find a definite clause in the program $\Pscr$ of the form
$\forall \bar u.\, G \imp p\ s_1\ \dotsm\ s_n$.
Let $S'$ be the result of replacing the guarded atom $p\ t_1\ \dotsm\ t_n$ in
$S$ with the following goal, then normalizing:
${
  \exists \bar u.\, \mleft({
    (t_1 = s_1) \wedge \dotsm \wedge (t_n = s_n) \wedge G
  }\mright)
}$.
Then the backchaining transition is $\twoseq{\Pscr}{S \transto[\epsilon] S'}$
(where $\epsilon$ is the empty substitution).\footnotemark

\footnotetext{%
  A backchaining step for the logic program of Figure~\ref{fig:eqlj-mod} is
  shown in
  Appendix~\ref{apx:backchaining-example}.
}

\paragraph{Unification steps:}

Let $t = s$ be a guarded equality in the reduced state formula $S$.
Reasoning about such terms uses a slight variant of Huet's pre-unification
procedure~\cite{huet75tcs}.
In the simplest case, if $t$ and $s$ are syntactically identical, then $S'$ is
just $S$ with the goal formula containing the guarded equality replaced with
$\true$, and the transition is $\twoseq{\Pscr}{S \transto[\epsilon] S'}$.

Following Huet, say that a term is \emph{rigid} if it is of the form
$(f\ t_1\ \dotsm\ t_n)$ where $f$ is a constant (called its \emph{head}).
If both $t$ and $s$ are \emph{rigid} terms with the same head, i.e., of the form
$f\ t_1\ \dotsm\ t_n$ and $f\ s_1\ \dotsm\ s_n$ respectively (for some constant
$f$), then we obtain $S'$ from $S$ by replacing the guarded equality $t = s$
with the following and normalizing:
$(t_1 = s_1) \wedge \dotsm \wedge (t_n = s_n)$ (or just $\true$ if $n = 0$).
The transition is $\twoseq{\Pscr}{S \transto[\epsilon] S'}$.

If both $t$ and $s$ are rigid terms with different heads, that particular
guarded equality is false, so $S'$ is produced from $S$ by replacing $t = s$
with $\false$, with the transition $\twoseq{\Pscr}{S \transto[\epsilon] S'}$.
Note that this does not mean that the $S$ itself is unsolvable.
For example, $\exists x.\, a = x \imp (f\ a) = (g\ b)$ has a solution
(instantiate $x$ with $b$) even if $f$ and $g$ are different constants.
We can thus replace this guarded goal with $\exists x.\, a = x \imp \false$ but
not with just $\false$.

The only other possible step is if one of $t$ or $s$ is a rigid term of the form
$f\ t_1\ \dotsm\ t_n$ (for some constant $f$), while the other is a
\emph{flexible} term of the form $x\ s_1\ \dotsm\ s_m$ where $x$ is an
existentially quantified variable of type
$\sigma_1 \ra \dotsm \ra \sigma_m \ra \tau$ in $S$, with the $\sigma_i$ and
$\tau$ being primitive types.
Thus, $S$ has the form $\Qscr_1.\, \exists x.\, \Qscr_2. C$, where $\Qscr_1$ and
$\Qscr_2$ are sequences of existential quantifiers and $C$ is the underlying
conjunction of guarded goals.
There are two possible kinds of transitions here.

\begin{itemize}
\item \emph{Imitation:} the flexible term is made to imitate the rigid term
  by setting $\rho$ to be the substitution (where each of the $x_i$ occur
  nowhere in $S$):
  \begin{gather*}
    x\mapsto \lambda w_1, \dotsc, w_m.\, f\, (x_1\ w_1\ \dotsm\
    w_m)\ \dotsm\ (x_n\ w_1\ \dotsm\ w_m).
  \end{gather*}
  The transition is
%  $\twoseq{\Pscr}{S \transto[\rho] \exists x_1, \dotsc, x_n.\, \rho[S]}$
  $\twoseq{\Pscr}{S \transto[\rho] \Qscr_1.\ \exists x_1, \dotsc, x_n.\, \Qscr_2.\,
    (C\rho)}$ 
\item \emph{Projection:} here, we set $\rho$ to be
  $x \mapsto \lambda w_1, \dotsc, w_m.\, w_i$ for some $i \in \{1, \dotsc, m\}$
  such that $\sigma_i = \tau$.
  %
%  The transition is $\twoseq{\Pscr}{S \transto[\rho] \rho[S]}$.
  The transition is $\twoseq{\Pscr}{S \transto[\rho] \Qscr_1\Qscr_2.(C\rho)}$.
\end{itemize}

\begin{theorem}[correctness of state transitions] \mbox{}
  \begin{itemize}
  \item (Soundness) if $\twoseq{\Pscr}{S \transto[\rho] S'}$ and $\theta$ is a
    solution for $\twoseq{\Pscr}{S'}$, then $\rho \circ \theta$ is a solution
    for $\twoseq{\Pscr}{S}$.
  \item (Relative completeness) if $\twoseq{\Pscr}{S}$ has a solution $\theta$,
    then there exists $\rho$ and $S'$ with $\twoseq{\Pscr}{S \transto[\rho] S'}$
    such that there is a solution $\theta'$ for
    $\twoseq{\Pscr}{S'}$ such that $\theta = \rho \circ \theta'$.
  \end{itemize}
\end{theorem}

\begin{proof}[Sketch]
  The proof follows the general outline of the proof of Theorem 8.2
  in~\cite{miller91jlc} and the correctness of the MATCH tree in Section 4
  of~\cite{huet75tcs}.
\end{proof}

Let $\transto*$ be the reflexive-transitive closure of $\transto$; i.e., we set
$\twoseq{\Pscr}{S \transto*[\epsilon] S}$ and we write
$\twoseq{\Pscr}{S \transto*[\rho] S'}$, if there is an interleaved sequence of
state formulas and substitutions $S_0, \rho_1, S_1, \rho_1, \dotsc, \rho_n, S_n$
with $S_0$ being $S$, $S_n$ being $S'$,
$\rho = \rho_1 \circ \dotsm \circ \rho_n$, and
$\twoseq{\Pscr}{S_{i - 1} \transto[\rho_i] S_i}$ for each
$i \in \{1, \dotsc, n\}$.
We say that a state formula $S$ is \emph{suspended} if there is no $\rho$ and
$S'$ such that $\twoseq{\Pscr}{S \transto[\rho] S'}$.

Given these rules for state transitions, our non-deterministic proof search
algorithm to prove the goal formula $G$ from a program $\Pscr$ involves first
converting $G$ to a reduced state formula $S$ and then searching for a state
formula $S'$ such that $\twoseq{\Pscr}{S\transto*[\rho] S'}$ and $S'$ is a
suspended state formula.  While these suspended state formulas may be complex
and difficult to determine their complete sets of solution, in practice, we
expect that most suspended state formulas will have simple structure.  This
situation is rather similar to Huet's pre-unification procedure which reduces
higher-order unification problems to sets of flexible-flexible disagreement
pairs: determining whether or not such sets of pairs have \emph{closed}
solutions is theoretically undecidable \cite[Section 8]{miller92jsc} but, in
practice, they are generally simple and seldom a concern.

The transition system outlined here is only a baseline and many optimizations
are possible.
For example, the occurs-check is not part of the transition system: if, say,
$x = f\ x$ is a target of a guarded equation (for an existentially quantified
$x$), then the imitation produces an infinite path.
Employing a notion of rigid-path checking for existentially quantified formulas
is worth doing.
The pattern fragment~\cite{miller91jlc} and the FCU subset of higher-order
unification~\cite{libal22amai} might be applicable to suspended states.
There can be other odd situations too: for example, a suspended guarded equation
of the form $t_1=t_2 \imp t_2=t_3\imp t_1=t_3$ is clearly solvable.

\section{Implementing an interactive search tactic in Abella}
\label{sec:abella}

As mentioned earlier, one of the goals of this work was to improve the
automation in the Abella theorem prover~\cite{baelde14jfr}.
While Abella has had a \lsti|search| tactic that does some basic search steps on
\emph{goals}, it is severely limited in its treatment of the full logic and of
rules such as ${=}L$ in general.

\subsection{Encoding the search problem}
\label{sec:abella.encoding}

Encoding the formulas of Abella in SLIM is generally straightforward along the
lines of the encoding of \eqLJone shown in Figure~\ref{fig:eqlj-mod}.
The only kind of formula that is not directly representable in SLIM is the
$\nabla$-quantifier~\cite{miller05tocl}, which we leave for future work.
The search problem is then a simple matter of packaging the proof context
consisting of the lemmas in scope and the local hypotheses in the form of a SLIM
state formula using a \lsti|seq| predicate.

\paragraph{Definitions yield two implications:}

Abella supports inductive and co-inductive definitions of atomic predicates,
which can be abstractly seen as a definition of a \emph{head} of the form
$p\ x_1\ \dotsm\ x_n$ in terms of a body formula $B$ whose free variables are
contained in $\{x_1, \dotsc, x_n\}$, with both head and body assumed to be
universally closed over these variables.
Such a definition can be translated into a pair of quantified implications,
namely:
%
% \begin{gather*}
${
  \forall \bar x.\, \mleft(p\ \bar x \imp B\mright)
}$
and
${
  \forall \bar x.\, \mleft(B \imp p\ \bar x\mright).
}$
%\end{gather*}

Both of these implications can be included (after encoding) among the collection
of assumptions as an argument to the \lsti|seq| predicate.
When encoded this way, the feature of being the least fixed point (for inductive
definitions) and the greatest fixed point (for coinductive definitions) is
lost.
Thus, this encoding is sound but not complete.

\paragraph{Inductive and co-inductive sizes:}

Abella implements induction and co-induction by a kind of guarded circular
reasoning, where the theorem statement is allowed to be used as a lemma in the
proof of an inductive or co-inductive theorem as long the inductive or
co-inductive argument can be shown to be strictly smaller or larger
(respectively).
This is achieved syntactically in Abella by means of annotations on atomic
formulas.
For example, for an inductively defined predicate $p$, an Abella theorem works
with two annotated versions $p\samesize$ (same size) and $p\smallersize$
(strictly smaller), with $p\samesize$ defined in terms of $p\smallersize$.
This can be achieved with a modification of the above pair of implications to
track the annotations.
Co-induction is treated similarly with a different pair of annotations.

\subsection{Bounding proof search}
\label{sec:abella.control}

Since both normalization and reduction processing are non-deterministic and
terminating, there are only two ways in which the transition system of
Section~\ref{sec:transition} can potentially loop.
First, the use of backchaining steps can be iterated indefinitely for
recursively specified clauses such as \lsti|seq| and \lsti|mnr| of
Figure~\ref{fig:eqlj-mod}.
To gain some control over such steps, one can create an annotated version of
these predicates with an extra size bound argument.
A much more promising alternative is to use a predicate such as
\lsti|seq Gamma G Xi| where the \lsti|Xi| argument is a \emph{foundational proof
  certificate (FPC)}~\cite{chihani17jar}.
An FPC can serve as a proof outline (with existentially quantified holes) for
which the logic programming engine must attempt to elaborate into a complete
proof term.
In this way, an FPC can be seen as a means to greatly control proof search. The
papers~\cite{blanco17cade,chihani15tableaux,miller24tplp} report on applications
of FPCs to support both proof search and proof reconstruction.

It is also important to bound the unification steps of the transition system,
because in the general case imitation steps can be non-terminating.
This is particularly the case since the unification steps do not implement an
occurs-check on guarded equalities, and thus a state formula such as
$\exists x.\, x = f\ x$ would explore infinitely deep terms.
However, there are other sources of non-termination, so a practical
implementation would require external bounds on the number of unification steps
before backtracking.

\section{Future and related work}

Equality is much more commonly considered within a \emph{classical logic}
setting, most notably through the theory of paramodulation~\cite{robinson83sc}.
Similarly, Comon and Lescanne~\cite{comon89jsc} explored \emph{disunification}
in a classical context, focusing on ``systems'' composed of conjunctions of
equations and \emph{disequations} $t \neq s$.
In our framework, such disequations are equivalent to $(t=s) \supset \false$;
however, our notion of states is fundamentally different since it is based on
intuitionistic logic.

The technique of modeling unification as a state transition system with
suspensions is well-established, appearing in an early paper of
Huet~\cite{huet73ijcai} and within the \emph{constraint logic programming} (CLP)
paradigm~\cite{jaffar87popl,tassi25coqpl}.
It is also standard in type and term inference systems, e.g.,~\cite{pientka13jfp}.

Future research will further investigate the parallels between the
current framework and CLP, potentially broadening the definition of a
state to encompass more general constraints.
While the reduced state formulas introduced here are essential for the
general problem, specific sub-cases may permit greater automation via
techniques such as \emph{constraint handling rules
(CHR)}~\cite{fruhwirth98jlp}.
Implementing the non-deterministic interpreter presented here---for
instance, as an extension of Abella---will necessitate the design of
explicit search strategies.
Furthermore, the interpreter could be enhanced with additional
proof-theoretic features, such as the
$\nabla$-quantifier~\cite{miller05tocl}, to handle generic judgments
more natively.

\section{Conclusion}
\label{sec:conclusion}

This paper addresses the challenge of integrating syntactic equality, treated as
a logical connective, into automated proof search.  While traditional frameworks
like \lP and LF lack direct support for this equality, we present a procedure
that handles the complexities of quantifier alternation and equations in both
positive and negative occurrences.  To facilitate automation, we introduce SLIM
(Specification Logic based on an Intuitionistic Metalogic), a lightweight
framework that enables the specification of diverse object-level sequent calculi
that may now involve syntactic equality.

The core of the proposed search procedure relies on transforming goal formulas
into state formulas through raising and normalization.  These state formulas are
then processed via a non-deterministic transition system involving backchaining
and a generalization of Huet's pre-unification procedure that deals with rigid
and flexible terms and various kinds of suspended equalities. Ultimately, this
system enables unification-aware proof search across various first-order
calculi, providing a robust foundation for future integration into theorem
provers like Abella.

\bibliographystyle{splncs04}
\bibliography{extract,local}

\appendix

\section{Provability from the empty logic program}
\label{apx:empty-lp}

Consider the simple case of the empty logic program and the two classes of
(first-order), atom-free formulas defined by the following recursive definition
(taken from \cite{miller22amai}):
\begin{align*}
  \Phi ::=&\
  \Phi_1 \wedge \Phi_2 \sep
  \exists_\sigma x.\, \Phi \sep
  \forall_\sigma x.\, \Phi \sep
  \Psi\\
  \Psi ::=&\
  t_1 =_{\sigma_1} s_1 \imp \dotsm \imp t_n =_{\sigma_n} s_n \supset t_0 =_{\sigma_0} s_0
\end{align*}
with the further assumption that every type $\sigma$ or $\sigma_i$ is a
primitive type.
The formulas in the syntactic class $\Phi$ are $G$ formulas, and they contain no
occurrences of atomic formulas.
The formulas in Examples~\ref{ex:equiv=} and~\ref{ex:succ} are examples of
$\Phi$-formulas.
We have the following propositions concerning the provability of such formulas
with respect to the empty logic program.

\begin{proposition}
  The provability of $\Pref[\cdot]\twoseq{\cdot}{\Phi}$ is undecidable, although
  if $\Phi$ contains no existential quantifiers, then provability is decidable.
\end{proposition}

\begin{proof}
  A reduction to Hilbert’s Tenth Problem (regarding finding solutions to
  Diophantine equations) was used in~\cite{miller22amai} to show that, in
  general, provability of the sequent $\Pref[\cdot]\twoseq{\cdot}{\Phi}$ is
  undecidable.
  Assume, however, that $\Phi$ contains no existential quantifiers.
  Note that in any \eqLJomega proof of $\Pref[\cdot]\twoseq{\cdot}{\Phi}$, any
  occurrence of the ${=}L$ rule involves only first-order unification (given
  that eigenvariables are always of primitive type).
  The result follows then from the fact that the right-introduction rules for
  $\forall$, $\wedge$, and $\imp$ and the left-introduction rule for equality are
  all invertible.
\end{proof}

\section{Example of backchaining in SLIM}
\label{apx:backchaining-example}

In the logic program of Figure~\ref{fig:eqlj-mod}, let $S$ be the state formula
\begin{gather*}
  \exists x.\,
  \forall y.\,
  (y = x~y) \imp \lsti|seq|\ (p\ y \lsti|::| \lsti|nil|)\ (\lsti|exists|~\lambda w.p~w)
\end{gather*}
Note that this state formula is reduced, and that it has two kinds of solutions,
namely, the one that instantiates $x$ with $\lambda u.\, u$ and the other that
instantiates $x$ with $\lambda u.\, t$ for any closed term $t$.
The backchaining state transition can be applied with any of the 17 clauses in
Figure~\ref{fig:eqlj-mod}.
Here, we pick the one that encodes the introduction of existential
quantification on the right: this clause can be written also as
\begin{gather*}
  \forall \Gamma.\, \forall C. (\exists x.\, \lsti|seq|\ \Gamma\ (C\ x))
  \imp (\lsti|seq|\ \Gamma\ (\lsti|exists|\ C)).
\end{gather*}
The result of backchaining with this goal first yields the goal formula 
\begin{align*}
  \exists x.\,
  \forall y.\,
  (y = x~y) \imp
  \exists \Gamma.\, \exists C.\,\bigl(&
  (p\ y \lsti|::| \lsti|nil|) = \Gamma) \wedge {}\\
  & (\lsti|exists|\ (\lambda w.\, p\ w)) = (\lsti|exists|\ C) \wedge {}\\
  & (\exists x.\, \lsti|seq|\ \Gamma\ (C\ x))\bigr).
\end{align*}
The result of reducing that goal formula yields the reduced state formula
\begin{align*}
  \exists x.\,
  \exists \Gamma'.\,
  \exists C'.\,
  \exists x'.\,
  \forall y.\,\Bigl(&
  \bigl(\forall y.\, y = x\ y \imp p\ y \lsti|::| \lsti|nil| = \Gamma'\ y \bigr) \wedge {}\\
  & \bigl(\forall y.\, y = x\ y \imp \lsti|exists|\ (\lambda w.\, p\ w) = \lsti|exists|\ (C'\ y)\bigr) \wedge {}\\
  & \bigl(\forall y.\, y = x\ y \imp \lsti|seq|\ \Gamma\ (C'\ (x'\ y)) \Bigr).
\end{align*}
Here, the quantified variables $\Gamma'$, $C'$, and $x'$ are raised versions of
the corresponding non-primed variable.  This state formula has two guarded
equalities and one guarded atom.

\section{One-sided MALL as an object logic}
\label{apx:linear}

As an illustration of the expressivity of SLIM, we give here a logic program
that specifies \emph{classical} first-order multiplicative-additive linear logic
with equality (MALL) in a one-sided formulation.

\begin{lstlisting}[language=lprolog,basicstyle={\ttfamily}]
kind tm, atm          type.             % Terms and atoms
kind fm               type.             % Formulas
type patom, natom     atm -> fm.        % Atoms (+/-)
type tens, par        fm -> fm -> fm.   % Multiplicative
type one, bot         fm.               % Mult. units
type with, plus       fm -> fm -> fm.   % Additive
type top, zero        fm.               % Add. units
type eq, neq          tm -> tm -> fm.
type forall,exists    (tm -> fm) -> fm. % quantifiers

kind fmlist  type.
type nil     fmlist.
type ::      fm -> fmlist -> fmlist. % written infix

type mnr     fm -> fmlist -> fmlist -> oo.
mnr X (X :: L) L.
mnr X (Y :: L) (Y :: K) :- mnr X L K.

type app     fmlist -> fmlist -> fmlist -> oo. % "append"
app nil L L.
app (X :: J) K (X :: L) :- app J K L.

type seq     fmlist -> o.

% exchange
seq (X :: Del) :- mnr X Gam Del, seq Gam.

% initial
seq (patom A :: natom A :: nil).

% multiplicative
seq (one :: nil).
seq (tens A B :: Gam) :-
  app DelA DelB Gam, seq (A :: DelA), seq (B :: DelB).
seq (bot :: L) :- seq L.
seq (par A B :: Gam) :- seq (A :: B :: Gam).

% additive
seq (top :: Gam).
seq (with A B :: Gam) :- seq (A :: Gam), seq (B :: Gam).
seq (plus A B :: Gam) :- seq (A :: Gam).
seq (plus A B :: Gam) :- seq (B :: Gam).

% equality
seq (eq T T :: Gam) :- seq Gam.
seq (neq S T :: Gam) :- (S = T) => seq Gam.

% quantifiers
seq (forall A :: Gam) :- pi x\ (A x :: Gam).
seq (exists A :: Gam) :- sigma t\ (A t :: Gam).
\end{lstlisting}

\end{document}